\def\soda{0}

\documentclass{article}

\usepackage{fullpage,amsthm}
\usepackage[disable]{todonotes}
\usepackage{graphicx} % support the \includegraphics command and options
\usepackage{array} % for better arrays (eg matrices) in maths
\usepackage{amsmath, amssymb, color, verbatim}
\usepackage{hyphenat,epsfig,subfigure,multirow}
\usepackage{hyperref}

\usepackage[ruled, noend, noline]{algorithm2e}

\newtheorem{theorem}{Theorem}[section]
\newtheorem{lemma}[theorem]{Lemma}
\newtheorem{proposition}[theorem]{Proposition}

\newtheorem{definition}{Definition}[section]

\newenvironment{proofof}[1]{\begin{proof}[of {#1}]}{\end{proof}}

\renewcommand{\qed}{\nobreak \ifvmode \relax \else
      \ifdim\lastskip<1.5em \hskip-\lastskip
      \hskip1.5em plus0em minus0.5em \fi \nobreak
      \vrule height0.75em width0.5em depth0.25em\fi}

\newcommand{\eps}{\epsilon}

\newenvironment{myproof}{\begin{proof}%\small
}{
\end{proof}}
\renewenvironment{proofof}[1]{\begin{myproof}[\ifnum\soda=0 Proof \fi of {#1}]%\small
}{\end{myproof}}

\begin{document}

\listoftodos

\title{Going for Speed: Sublinear Algorithms for Dense $r$-CSPs}
\author{Grigory Yaroslavtsev \footnote{This work was done while the author was supported by a postdoctoral fellowship at the Warren Center for Network and Data Sciences at the University of Pennsylvania and the Institute Postdoctoral Fellowship at the Brown University, ICERM. } \\ University of Pennsylvania \\ \tt{grigory@grigory.us}}

\maketitle

\begin{abstract}
We give new sublinear and parallel algorithms for the extensively studied problem of approximating $n$-variable $r$-CSPs (constraint satisfaction problems with constraints of arity $r$) up to an additive error $O(\epsilon n^r)$.
The running time of our algorithms is $O\left(\frac{n}{\epsilon^2}\right) + 2^{O\left(\frac{1}{\epsilon^2}\right)}$ \todo{Corrected running time.} for Boolean $r$-CSPs and $O\left(\frac{k^4 n}{\epsilon^2}\right) + 2^{O\left(\frac{\log k}{\epsilon^2}\right)}$ \todo{Corrected running time.} for  $r$-CSPs with constraints on variables over an alphabet of size $k$.
For any constant $k$ this gives optimal dependence on $n$ in the running time unconditionally, while the exponent in the dependence on $1/\epsilon$ is polynomially close to the lower bound under the exponential-time hypothesis, which is  $2^{\Omega(1/\sqrt{\epsilon})}$.\todo{Rephrasing.}

For \textsc{Max-Cut} this gives an exponential improvement in dependence on $1/\epsilon$ compared to the sublinear algorithms of Goldreich, Goldwasser and Ron (JACM'98) and a linear speedup in $n$ compared to the algorithms of Mathieu and Schudy (SODA'08).
For the maximization version of $k$-\textsc{Correlation Clustering} problem
 our running time is $O(k^4 n / \epsilon^2) + k^{O(1/\epsilon^2)}$\todo{Corrected running time.}, 
 improving the previously best $n k^{O\left(\frac{1}{\epsilon^3} \log \frac{k }{ \epsilon} \right)}$ by Guruswami and Giotis (SODA'06).
\end{abstract}

\section{Introduction}

Approximation algorithms for constraint satisfaction problems have received a lot of attention in the recent years. In particular, polynomial-time approximation schemes for dense instances have been developed using several different approaches~\cite{V96, GGR98, AKK99, FK99, AVKK03, VKKV05, VM07, MS08, AFNS09, BHHS11, YZ14}, including combinatorial sampling methods, subsampling from linear programming and semidefinite programming relaxations and rounding of hierarchies of linear programming relaxations.
Notably, some of these algorithms run in sublinear time in the input size~(e.g. \cite{GGR98, AVKK03}).  In particular, for \textsc{Max-Cut}, probably the most commonly studied CSP, a partition which gives a cut of size within $\epsilon n^2$ of the optimum can be constrained in linear in $n$ time for fixed $\epsilon$~\cite{GGR98}.

In this paper we revisit the problem of constructing approximate solutions for instances of $r$-CSPs with additive error $\epsilon n^r$, focusing on algorithms which run in sublinear time. 
To the best of our knowledge, even for \textsc{Max-Cut}, the most basic $2$-CSP, among all algorithms considered in the previous work only those of~\cite{GGR98} can be used to reconstruct the solution in sublinear time. We are not aware of any algorithms for general $r$-CSPs over large alphabets which achieve sublinear running time\todo{Rephrashing}. Note that multiple algorithms exist (e.g.~\cite{AVKK03, MS08}) for approximating the cost of the optimum solution in sublinear (and even constant) time, however all these approaches take at least linear time to reconstruct the solution itself (see Table~\ref{table:results}).
 The motivation for fast algorithms comes from applications of $r$-CSPs to areas such as clustering and graph partitioning. For example, the results of~\cite{GG06} can be interpreted as showing (implicitly) that \textsc{Correlation Clustering} into $k$ clusters can be expressed as a $2$-CSP problem over an alphabet of size $k$ and hence algorithms for approximating CSPs can be directly applied to this problem.\todo{Rephrasing.}

\subsection{Previous work}
In Table~\ref{table:results} we compare the running times of our algorithms against those achieved prior to this work via three different approaches to designing PTASes for dense $r$-CSPs: combinatorial algorithms based on sampling, subsampling of mathematical relaxations and rounding of linear programming hierarchies.
For any fixed $k$ our algorithms give the best running time.
\todo{Rephrasing.}

\paragraph{Combinatorial algorithms with sampling.}
Historically this is the first framework which led to development of PTASes for dense problems.
For \textsc{Max-Cut} and some other graph problems sampling-based combinatorial algorithms were developed by de la Vega~\cite{V96}, Goldreich, Goldwasser and Ron~\cite{GGR98} and Arora, Karger and Karpinski~\cite{AKK99}. Prior to our work, which falls into this category, fastest algorithms for dense problems were also obtained via this framework. Fastest of such algorithms were~\cite{GGR98} and~\cite{MS08}, which have incomparable running times. Former work achieves sublinear dependence on the input size, while its dependence on $\epsilon$ is worse. Latter work gives algorithms with linear running time in the input size, while achieving better dependence on $\epsilon$. 

Efficiency of combinatorial algorithms based on sampling makes them appealing for applications to clustering and graph partitioning~\cite{GGR98, GG06}. Some of these algorithms~\cite{GGR98} are extremely easy to parallelize, while others~\cite{MS08} are inherently sequential. Our work falls in between these two extremes --- the basic versions of our algorithms are sequential but we show how to transform them into parallel algorithms. We discuss parallel versions of our algorithms in Section~\ref{sec:conclusions}.

\paragraph{Subsampling of mathematical relaxations.}
This approach was first suggested by~\cite{AVKK03}, who studied the effect of subsampling on the value of the objective function of mathematical relaxations of dense NP-hard problems.
It was further extended to more general settings by~\cite{BHHS11}. To the best of our knowledge, all algorithms obtained in this line of work run either in linear or polynomial time in the input size.

\paragraph{Rounding linear programming hierarchies.} The most recent approach is based on rounding linear programming hierarchies. De la Vega and Kenyon-Matheiu~\cite{VM07} showed that $O(1 / \epsilon^2)$ rounds of Sherali-Adams hierarchy suffice to bring the integrality gap for dense instances of \textsc{Max-Cut} down to $(1 + \epsilon)$. More generally, Yoshida and Zhou~\cite{YZ14} showed that $2^{O(r)} \log k / \epsilon^2$ rounds of Sherali-Adams suffice for general $r$-CSPs over an alphabet of size $k$.

%\begin{itemize}
%\item \cite{GGR98}, who give algorithms for \textsc{Max-Cut} in time $2^{O\left(\frac{\log 1 / \epsilon }{ \epsilon^3} \right)} + n \cdot O\left(\frac{\log 1/ \epsilon }{ \epsilon^2}\right)$
%\item \cite{YZ14}, who give approximation algorithms for $k$-CSPs via Sherali-Adams.
%\item \cite{GG06} --- their algorithms have running time $n k^{O\left(\frac{1}{\epsilon^3} \log k / \epsilon\right)}$ for the maximization version of $k$-\textsc{Correlation Clustering}.
%\item \cite{MS08} Mathieu-Schudy, whose work we build upon.
%\end{itemize}

\begin{table}[t]
\begin{center}
\begin{tabular}{|l || c | c| c| c|}
\hline
\multicolumn{5}{|c|}{\bf Running times of approximation schemes for dense $r$-CSPs } \\
\hline
\hline

 & \textsc{Max-Cut}
 & Binary $r$-CSP ($r > 2$)
 & $k$-ary $2$-CSP 
 & $k$-ary $r$-CSP ($r > 2$)
    
    \\
      \hline
      \multicolumn{5}{|c|}{\bf Combinatorial algorithms with sampling } \\
             \hline
           \cite{AKK99}
           & $n^{O(1 / \epsilon^2)}$
           & ---
           & ---
           & ---
  
  \\
             \hline
           \cite{V96}
           & $O\left(n^2 \cdot 2^{1 / \epsilon^{2 + o(1)}}\right)$
           & ---
           & ---
           & ---
       \\
    \hline
      \cite{GGR98}
      & $ O\left(\frac{n \log 1 / \epsilon}{\epsilon^2}\right) + 2^{O\left(\frac{\log 1 / \epsilon}{\epsilon^3}\right)}$
      & ---
      & $O\left(\frac{n \log k / \epsilon}{\epsilon^2}\right) + 2^{O\left(\frac{\log^2 k / \epsilon}{\epsilon^3}\right)}$
      & ---
      
      \\
      \hline
      \cite{MS08}
      & 
      $O(n^2) + 2^{O(1/\epsilon^2)}$
      & $O(n^r)+ 2^{O(1 / \epsilon^2)}$
      & $O(n^2) + 2^{O\left(\frac{\log k}{\epsilon^2} \right) }$
      &  $O(n^r) + 2^{O\left(\frac{\log k}{\epsilon^2} \right) }$
      \\
      \hline
      Our work
      & $ O\left(\frac{n }{ \epsilon^2}\right) + 2^{O\left(\frac{1 }{\epsilon^2}\right)}$ 
      & $O\left(\frac{n}{ \epsilon^2}\right) + 2^{O\left(\frac{1}{\epsilon^2}\right)}$
      & $O\left(\frac{n k^4 }{ \epsilon^2}\right) + 2^{O\left(\frac{\log k }{\epsilon^2}\right)}$
      & $O\left(\frac{n k^4 }{ \epsilon^2}\right) + 2^{O\left(\frac{\log k}{\epsilon^2}\right)}$
    \\
\hline
\multicolumn{5}{|c|}{\bf Subsampling of mathematical relaxations } \\

                 \hline
               \cite{AVKK03}
               & $O(n^2) \cdot 2^{\tilde O(1 / \epsilon^2)}$
               & $O(n^r) \cdot 2^{\tilde O(1 / \epsilon^2)}$
               & ---
               & ---
  \\
  \hline
  \multicolumn{5}{|c|}{\bf Rounding linear programming hierarchies } \\
        \hline
              \cite{VM07}
              & $n^{O(1  / \epsilon^2)}$
              & ---
              & ---
              & ---
              \\
        \hline
      \cite{YZ14}
      & $n^{O(1  / \epsilon^2)}$
      & $n^{2^{O(r)}  / \epsilon^2}$
      & $n^{O(\log k  / \epsilon^2)}$
      & $n^{2^{O(r)}  \log k / \epsilon^2}$
\\  
\hline
\end{tabular}
\end{center}
\caption{Approximation algorithms with additive error $\epsilon n^r$.}
\label{table:results}
\end{table}

\subsection{Related work in streaming algorithms}

We would like to highlight the connection between the line of work on sublinear time algorithms and related area of sublinear space streaming algorithms.
Recently it has been shown that $(1 \pm \epsilon)$-approximate \textit{cut sparsifiers} and \textit{spectral sparsifiers} can be constructed in the streaming model in $\tilde O(n / \epsilon^2)$ space~\cite{AGM12, GKP12, KL13, KLMMS14}.
This is almost optimal (see~\cite{BK96, ST11, SS11, BSS14}) and in particular implies that dense \textsc{Max-Cut} problem that we consider can be solved in almost optimal sublinear space in the streaming model, if there are no restrictions on the running time. The dependence on $n$ in our algorithms for \textsc{Max-Cut} is sublinear and optimal, as is space in the streaming algorithms. It remains open whether both of these sublinear time and space bounds can be achieved by the same algorithm. 

\subsection{Our results and techniques}

In order to introduce the main ideas we first develop a simpler version of our algorithm for \textsc{Max-Cut} in Section~\ref{sec:max-cut-warmup}. Our main algorithm for \textsc{Max-Cut} follows uses a more general approach that we take in Section~\ref{sec:main} and has the following guarantee.
\begin{theorem}\label{thm:max-cut}
There is an algorithm which approximates \textsc{Max-Cut} with additive error $O(\epsilon n^2)$ and runs in time $O(n / \epsilon^2) + 2^{O(1/\epsilon^2)}$.
\end{theorem}

The intuition behind our algorithms comes from~\cite{GGR98} and~\cite{MS08}. 
In particular, our algorithm for \textsc{Max-Cut} is inspired by~\cite{GGR98}, who introduced a sampling-based technique for approximating dense instances of \textsc{Max-Cut} with sublinear running time. However, we depart from their approach, which partitions the graph into $\Theta(1 / \epsilon)$ parts and then partitions each part based on an optimum partition of a sample of size $\Theta(1  / \epsilon^2)$ drawn from the rest of the graph. Such an approach seems to inherently require the complexity to be $2^{\Omega(1 / \epsilon^3)}$ due to the partitioning step.
%It is natural to suggest that maybe a single sample of size $O(1 / \epsilon^2)$ would suffice for the entire graph. However, turns out that this would be insufficient to obtain a good enough approximation. 
Instead, we use a bootstrapping scheme from the linear time algorithm of~\cite{MS08}. In a nutshell, an optimum solution on a primary sample of size $O(1 / \epsilon^2)$ suffices to partition a large secondary sample of size $O(1 / \epsilon^4)$ within an additive error $O(1 / \epsilon^7)$ and then this approximate solution on the secondary sample can be further used to construct an approximate solution for the entire graph. The original algorithm of~\cite{MS08} takes $O(n^2)$ time, i.e. linear in the input size to convert the approximately optimal solution for the secondary sample into an approximately optimal solution for the entire graph. We show how to choose sampling rate at each step of the algorithm in order to maintain the same approximation guarantee in optimal sublinear time.

 \todo{Extended discussion of new ideas as compared to~\cite{MS08} as per Comment 1.} 
We introduce the main ideas in Section~\ref{sec:max-cut-warmup}, where we describe a simplified version of our algorithm for \textsc{Max-Cut}, Algorithm~\ref{alg:greedy-ptas-sampling}. If the greedy step in this algorithm is performed using exact degrees then the analysis of this algorithm is given in~\cite{MS08}.
While the idea of using sampling in the greedy step might seem natural, the details are quite involved. Indeed, a single greedy step can be easily seen to introduce only small error even if it is made based on approximate degrees. However, Algorithm~\ref{alg:greedy-ptas-sampling} executes the greedy steps sequentially and degrees used to place the current vertex depend on the placement of the previous vertices. This may allow the errors the algorithm makes when placing vertices to affect degrees used when placing subsequent vertices and thus lead to amplification of errors. We use martingale-based analysis based on~\cite{MS08} to break down the error analysis in such a way that it can be done independently for every step (Lemma~\ref{lem:step-bound-aux}). It is crucial that this analysis still applies to the new martingale that we use, which tracks performance of an approximate greedy algorithm instead of an exact one.
Given this, there are two sources of error in every step: the fact that we are using a greedy choice and the approximation used in this choice.
In Lemma~\ref{lem:step-bound} we show that these two sources can be treated independently and analyze the error introduced by sampling in Proposition~\ref{prop:sampling-concentration}.
To analyze the error introduced by the greedy choice we use the martingale-based argument of~\cite{MS08} adapted to our new martingale (Lemma~\ref{lem:additive-error}). It is crucial that the argument is robust to the change in the definition of the martingale -- we show this by reproducing the analysis.
Finally, we give the overall analysis of the approximation given by Algorithm~\ref{alg:greedy-ptas-sampling} in Lemma~\ref{lem:key-lemma}. In order to achieve the desired additive bound together with linear running time of the greedy step we choose the sample size to be time-dependent. Selection of the sampling rate is one of the most challenging parts of the analysis. 

In Section~\ref{sec:main} we give the analysis of our main algorithm, Algoirthm~\ref{alg:ptas-rcsp}, which is a faster version of Algorithm~\ref{alg:greedy-ptas-sampling} and also works for general $r$-CSPs.
It introduces the bootstrapping step in order  to reduce the running time.
This proof follows a similar structure and hence all observations above apply.
Again, we are able to break the analysis into steps and separate the error introduced by the greedy selection process and the approximation involved in it. Even though the details are more delicate since we now have to deal with the more general case of $r$-CSPs we are still able to adapt some of the technical lemmas of~\cite{MS08} under the new definitions that we use.
This leaves us two main technical challenges to address: selection of the sample space and choice of the sampling rate sufficient to achieve the desired approximation. For \textsc{Max-Cut} the sampling space is easily determined to be the set of already placed neighbors of the vertex, which is currently being placed greedily. For general $r$-CSPs we show that an appropriate generalization is the set of all \textit{critical constraints} (see Definition~\ref{def:critical-constraint}) containing currently processed variable. This choice is crucial for the analysis to be extendable to this case because it allows to bound the size of the sampling space at time $t$ by $t^{r - 1}$. It is important that the sampling space can be shown to be much smaller than the set of all constraints involving current variable, which may have size $n^{r -1}$. Together with a careful choice of time-dependent sampling rate, which generalizes the one we use in Algorithm~\ref{alg:greedy-ptas-sampling}, the bound on the size of the sampling space allows to achieve both the optimal bound on the running time and the additive approximation guarantee.

For general $r$-CSPs over an alphabet of size $k$ we obtain the following results, which follow from the analysis of Algorithm~\ref{alg:ptas-rcsp}.

\begin{theorem}\label{thm:rcsp}
There is an algorithm which approximates any $r$-CSP problem over alphabet of size $k$ within additive error $O(\epsilon n^r)$ and runs in time $O\left(\frac{n k^4}{ \epsilon^2}\right) + 2^{O\left(\frac{\log k}{\epsilon^2}\right)}$.
\end{theorem}
Theorem~\ref{thm:rcsp} immediately implies algorithms with the same running time for the $k$-\textsc{Correlation Clustering} problem by taking $r = 2$.

Our last result is a lower bound, which is proved in Section~\ref{sec:lower-bound} and complements the performance guarantees of our algorithms, implying a lower bound of $\Omega(n / \epsilon^2) + 2^{\Omega(1 / \sqrt{\epsilon})}$ on the running time of any algorithm for \textsc{Max-Cut}, assuming ETH.

\begin{theorem}\label{thm:lower-bound}
Any algorithm $\mathcal A$, which approximates \textsc{Max-Cut} within error $\epsilon n^2$ in the adjacency list model has to make at least $\Omega(n / \epsilon^2)$ queries to the edges of the graph. Every such algorithm also has to have at least $2^{\Omega(1 / \sqrt{\epsilon})}$ running time assuming the exponential time hypothesis.
\end{theorem}

This theorem justifies the fact that the running time of sublinear algorithms for $r$-CSPs is naturally divided into terms, the first one corresponding to the query complexity and the second one corresponding to the computational complexity of the problem. While the first term in our work is provably unconditionally almost tight by Theorem~\ref{thm:lower-bound}, we can only lower bound the second term conditionally since it corresponds to the computational complexity of the problem. 

The computational lower bound in Theorem~\ref{thm:lower-bound} is easy to show. Consider \textsc{Max-Cut} on instances of size $1 / \sqrt{\epsilon}$. Assuming ETH, such instances can't be solved exactly in time $2^{o(1 / \sqrt{\epsilon})}$. However, an additive error guarantee of $\epsilon n^2$ requires that such instances have to be solved exactly by our algorithm. The query complexity lower bound comes from the intuition that any PTAS for dense instances of \textsc{Max-Cut} has to estimate the degrees of at least a constant fraction of vertices up to an additive error $\epsilon n$, which requires a sample of size $\Omega(1 / \epsilon^2)$ per vertex and yields the overall lower bound.
However, the technical details of this proof are more involved and are given in Section~\ref{sec:lower-bound}.
We use Yao's principle and construct a random family of hard instances as follows.
Let $V = V_0 \cup V_1 \cup V_2$, where $|V_0| = |V_1| = 4n/9$ and $|V_2| = n/9$. The vertex set $V_0 \cup V_1$ induces a complete bipartite graph $K_{4n/9, 4n/9}$ with parts $V_0$ and $V_1$, which corresponds to a planted dense solution for \textsc{Max-Cut}. For each vertex $v \in V_2$ we randomly pick a side of this cut $r_v\in \{0,1\}$ with probability $1/2$ each and add edges with probability $1/2 + \epsilon$ to each vertex on the side $r_v$ and with probability $1/2$ to each vertex on the other side $1 - r_v$.  Thus, in the optimum solution each vertex in $v \in V_2$ has to be placed on the side $1 - r_v$ of the cut.
The intuition behind the lower bound is that even if the algorithm guesses the planted solution on $V_0 \cup V_1$ without any queries then in order to get an additive error $c \epsilon n^2$ for some sufficiently small $c$ it still has to guess $1 - r_v$ with probability greater than $1/2$ for   vertices $v \in V_2$. This is impossible without sampling at least $\Omega(1 / \epsilon^2)$ edges by a Chernoff-type lower bound against sampling algorithms~\cite{CEG95}. However, the technical details are more complicated for two reasons. First, an approximation algorithm can partition $V_0 \cup V_1$ differently than the optimum planted solution, which might simplify the task of guessing the optimal side for vertices in $V_2$.
We show that unless the partitioning of $V_0 \cup V_1$ is sufficiently close to optimum the algorithm incurs an error of $\Omega(\epsilon n^2)$ on edges induced by these vertices alone.
Second and more subtle issue is the way $V_0 \cup V_1$ is partitioned by the algorithm might depend on the edges adjacent to vertices in $V_2$. To address this we use a probabilistic argument, arguing by a union bound that for \textit{every} partition of $V_0 \cup V_1$, which is sufficiently close to the optimum, placing vertices in $V_2$ optimally still requires $\Omega(n / \epsilon^2)$ queries.

\section{Max-Cut}\label{sec:max-cut-warmup}

To illustrate the main ideas we first present Algorithm~\ref{alg:greedy-ptas-sampling}, which demonstrates how sampling can be used to speed up an approximate greedy algorithm.

\begin{algorithm}
\caption{Greedy PTAS with subsampling.}\label{alg:greedy-ptas-sampling}
\SetKwInOut{Input}{input}\SetKwInOut{Output}{output}
\Input{Graph $G(V,E)$, where $|V| = n$, parameter $\eps$.}
%\Output{Accept/Reject}
\DontPrintSemicolon
\BlankLine

\nl Pick a sample $S$ of $t_0 = 1/\epsilon^2$ vertices uniformly at random without replacement\;
\nl \For {each of the $2^{t_0}$ possible partitions of $S$ into two parts}
{ \nl $S^{t_0} = S, t = t_0 + 1$ \;
\nl \For {each vertex $v \in V \setminus S$ in random order} 
{ \nl Pick a sample $V^t$ of $s_t = O\left(\frac{n^{2/3}}{t^{2/3}\epsilon^2}\right)$ vertices uniformly at random without replacement from $S^{t - 1}$.  \; 
\nl Assign $v$ to the side of the cut, which maximizes the number of cut edges with respect to the current partition of $V^t$ \;
\nl $S^t = S^{t - 1} \cup \{v\}, t = t + 1$ }}
\nl Output the best cut over all iterations\;
\end{algorithm}
\todo{Added time-dependent sampling rate as per Comment 4.}

\begin{theorem}\label{thm:max-cut-greedy-sampling}
Algorithm~\ref{alg:greedy-ptas-sampling} gives an additive $O(\epsilon n^2)$-approximation for \textsc{Max-Cut} in time $O(n) \cdot 2^{O(1/\epsilon^2)}$
\end{theorem}

%The algorithm selects a sample $S$ of size $t_0 = 1/\epsilon^2$ and for each of the $2^{t_0}$ possible cuts on the sample assigns the vertices in $V \setminus S$ using a greedy rule.
%Vertices in $V \setminus S$ are considered in a random order an placed on the side which maximizes the number of cut edges. The output is the best solution obtained over all iterations.

First, note that the total size of all samples is $\sum_{\tau = t_0}^n s_\tau = \frac{n^{2/3}}{\epsilon^2} \sum_{\tau = t_0}^n \tau^{- 2/3} = O\left( \frac{n }{ \epsilon^2}\right)$. Thus, the overall running time of the algorithm is $O\left(\frac{n}{\epsilon^2} 2^{1 / \epsilon^2}\right)$ for all iterations of the loop\footnote{For this analysis of the running time as well as for the analysis of the approximation below (Lemma~\ref{lem:key-lemma}) any choice of sample size $s_t = \frac{n^\delta}{t^\delta \epsilon^2}$ would suffice. The choice of $\delta = 2/3$ allows to minimize the constant factor in the running time.}. \todo{New running time analysis with time-dependent samping rate + remark on possible sample sizes.}

The sides of the cut are indexed by $i \in \{1,2\}$.
A cut is represented by a vector $x \in \{0,1\}^{2n}$, where $x_{ui} = 1$ iff the vertex $u$ is assigned a label $i$.
Let $A$ to denote a matrix with entries $A_{i,u_1,j, u_2}$ defined as follows:
$A_{i, u_1, j, u_2} = 1/2$ if $u_1 = u_2$ and $(i,j) \in E$ and $A_{i, u_1, j, u_2} = 0$ otherwise.
Then we can express the objective function as minimization of a bilinear form $x^T A x$.
For a fixed iteration $t$ of the algorithm we use $A^t$ to denote the matrix sampled from $A$ with columns corresponding to vertices in $V^t$ being the same as the columns of $A$, while all other columns replaced be zeros. Formally $A^t_{i, u_1, j, u_2} = A_{i, u_1, j, u_2}$ if $j \in V^t$ and $A^t_{i, u_1, j, u_2} = 0$ otherwise.

We will track the solution obtained at time $t$ using variables $x^t_{ui}$ such that $x^t_{ui} = 1$ if at time $t$ the vertex $u$ is assigned label $i$ and $x^t_{ui} = 0$ otherwise (either $u$ is assigned a different label or not assigned a label at all).
Let $r_t$ denote the $t$-th vertex considered by the algorithm.
Let $S^t$ denote the set of vertices assigned by time $t$.
Let $x^*$ be the optimum cut.
We denote the exact and approximate greedy choices at time $t$ for each vertex $u$ as $\tilde g^{t}_{ui}$ and $g^t_{ui}$, which are given as:
\begin{align*}
\tilde g^{t}_{ui} = 
\begin{cases}
x^*_{ui} \text{, if } t \le t_0, \\
1 \text{, if } t > t_0,  i = \arg \min_j A_{u j} x^{t - 1}\\
0 \text{, otherwise.}
\end{cases}
&&&
g^t_{ui} = 
\begin{cases}
x^*_{ui} \text{, if } t \le t_0, \\
1 \text{, if } t > t_0,  i = \arg \min_j A^t_{u j} x^{t - 1}\\
0 \text{, otherwise.}
\end{cases}
\end{align*}
By the definition of the greedy step we can write $g^t_{r_t} = x^t_{r_t}  - x^{t - 1}_{r_t}$.

A \textit{fictitious cut} is defined using a set of auxiliary variables $\hat x^t_v$ such that $\hat x^t_v = x^t_v$ if $v \in \{r_1, \dots, r_t\}$ and $\hat x^t_v = \frac{1}{t} \sum_{\tau = 1}^t g^\tau_v$ otherwise.

\begin{lemma}\label{lem:step-bound-aux}
For every $t$ it holds that:
\begin{align*}
(\hat x^t)^T A \hat x^t - (\hat x^{t - 1})^T A \hat x^{t - 1} \le 2 (\hat x^t - \hat x^{t - 1})^T A \hat x^{t - 1} + \frac{4 n^2}{t^2}.
\end{align*}
\end{lemma}
\begin{proof}
%Note that $b(x) = 2 A x$ and thus $x^T A x = \frac{1}{2} x^T b(x)$.
%Hence, we have:
% \begin{align*}
% z(\hat x^t) = (\hat x^{t - 1} + \hat s^t)^T A (\hat x^{t -1} + \hat s^t) = (\hat x^{t - 1})^T A \hat x^{t - 1} + 2  (\hat s^t)^T A \hat x^{t - 1} + (\hat s^t)^T A \hat s^t. 
% \end{align*}
%Thus, we have:
%\begin{align*}
%q(\hat x^t) - q(\hat x^{t - 1}) = \hat s^t b(x^{t - 1}) + (\hat s^t)^T A \hat s^t.
%\end{align*}

Observe that $(\hat x^t)^T A \hat x^t - (\hat x^{t - 1})^T A \hat x^{t - 1} = 2 (\hat x^t - \hat x^{t - 1})^T A \hat x^{t - 1} + (\hat x^t - \hat x^{t - 1})^T A (\hat x^t - \hat x^{t - 1}).$

In order to bound the second term let's express the components of $\hat x^t - \hat x^{t - 1}$.
There are three cases:
\begin{enumerate}
\item \textbf{Case 1.} $u \in S^{t}, u \notin S^{t - 1}$.
In this case we have $\hat x^t_u = x^t_u =  g^t_u$ by the definition of greedy. Also, $x^{t - 1}_u = 0$ because $u$ hasn't been assigned yet at time $t - 1$.
Thus, $\hat x^t_u - \hat x^{t - 1}_u = g^t_u - \hat x^{t - 1}_u$.
\item \textbf{Case 2.} $u \notin S^{t}$. In this case we have $\hat x^t_u = \frac{1}{t} \sum_{\tau = 1}^t g^\tau_u$ and $\hat x^{t - 1}_u = \frac{1}{t - 1} \sum_{\tau = 1}^{t - 1} g^\tau_u$
We have $\hat x^t_u - \hat x^{t - 1}_u = \frac{1}{t}g^t_u - \frac{1}{t (t - 1)} \sum_{\tau = 1}^{t - 1} g^\tau_u = \frac{1}{t} \cdot (g^t_u - \hat x^{t - 1}_u)$.
\item \textbf{Case 3.} $u \in S^{t - 1}$.
In this case $\hat x^t_u - \hat x^{t - 1}_u = 0$ because $\hat x^t_u = \hat x^{t - 1}_u = x^{t - 1}_u$.
\end{enumerate}
Thus, we have that $|\hat x^t - \hat x^{t - 1}|_1 = \sum_u |\hat x^t_u - \hat x^{t - 1}_u|_1 \le 2 + \frac{2}{t} (n - t) = 2 \frac{n}{t}$.
This implies that $(\hat x^t - \hat x^{t - 1})^T A (\hat x^t  - \hat x^{t - 1}) \le \max_{ij} A_{ij}  \cdot  |\hat x^t - \hat x^{t - 1}|^2_1 \le 4 \frac{n^2}{t^2}$, completing the proof.
\end{proof}

We denote $q^t = \hat x^t - \frac{n}{t} x^t$.

\begin{lemma}\label{lem:step-bound}
For all $t \ge t_0$ it holds that:
\begin{align*}
\mathbb E \left[(\hat x^t)^T A \hat x^t - (\hat x^{t - 1})^T A \hat x^{t - 1}\right] \le \frac{4 n^2}{t^2} + \frac{4 n^2}{(t - 1) \sqrt{s_t}} +  \frac{2n}{t (n - t + 1)} \mathbb E \left[ \left|A q^{t - 1}\right|_1\right]
\end{align*}
\end{lemma}
\begin{proof}
We first show the following auxiliary statement.

\begin{proposition}\label{prop:sampling-concentration}
For every $u \notin S^{t - 1}$ it holds that $\mathbb E \left[(g^t_{ui} - \hat x^{t - 1}_{ui})^T A x^{t - 1}\right] \le \frac{2t}{\sqrt{s_t}}$. \todo{Changed to reflect time-dependent sampling rate.}
\end{proposition}
\begin{proof}
We have 
\begin{align*}
(g^t_{ui} - \hat x^{t - 1}_{ui})^T A x^{t - 1} = (g^t_{ui} - \tilde g^t_{ui}) A x^{t - 1} + (\tilde g^t_{ui} - \hat x^{t - 1}_{ui})^T A x^{t - 1} \le (g^t_{ui} - \tilde g^t_{ui}) A x^{t - 1}
\end{align*}
where the inequality follows from the fact that the optimal greedy choice $\tilde g^t$ at step $t$ minimizes $\tilde g^t_{ui} A x^{t - 1}$ for every $u$ over the choice of $i$ and thus $(\tilde g^t_{ui} - \hat x^{t - 1}_{ui})^T A x^{t - 1} \le 0$.
W.l.o.g we assume that the optimal greedy choice is $\tilde g^t_{u1} = 1$, i.e. the the vertex $u$ is assigned label $i = 1$. Then for the term $(g^t_{ui} - \tilde g^t_{ui}) A x^{t - 1}$ we have:
\begin{align*}
(g^t_{ui} - \tilde g^t_{ui}) A x^{t - 1} =
\begin{cases}
\tilde N^t_{u1} - \tilde N^{t}_{u2}, &\text{if } N^t_{u1} < N^t_{u2} \\
0, &\text{if } N^t_{u1} \ge N^t_{u2}, 
\end{cases}
\end{align*} 
where $\tilde N^t_{ui}$ and $N^t_{ui}$ denote the number of neighbors of $u$ with label $i$ in $S^{t - 1}$ and $V^{t - 1}$ respectively.
By a union bound, we have: 
\begin{align*}
\Pr[N^t_{u1} < N^t_{u2}] \le \Pr\left[N^t_{u1} < \frac{s_t}{t - 1} \frac{\tilde N^t_{u1} + \tilde N^t_{u2}}{2}\right] + \Pr\left[N^t_{u2} > \frac{s_t}{t - 1} \frac{\tilde N^t_{u1} + \tilde N^t_{u2}}{2}\right]
\end{align*}
We introduce notation $\alpha = \frac{\tilde N^t_{u1} - \tilde N^t_{u2}}{2}$.
Since $\mathbb E\left[N^t_{u1}\right] = \frac{s_t}{t - 1} \tilde N^t_{u1}$ by an additive Hoeffding bound we have:
\begin{align*}
\Pr\left[N^t_{u1} < \frac{s_t}{t - 1} \tilde N^t_{u1} - \alpha \frac{s_t}{t - 1}\right] \le e^{- \frac{2 \alpha^2 s_t }{(t - 1)}}.
\end{align*}
Bounding the second term similarly we get that $\Pr[N^t_{u1} < N^t_{u2}] \le 2 e^{- \frac{2 \alpha^2 s_t}{(t - 1)^2}}$.
Thus, $\mathbb E[(g^t_{ui} - \tilde g^t_{ui}) A x^{t - 1}] \le 4 \alpha e^{- \frac{2 \alpha^2 s_t}{(t - 1)^2}}$. Taking the derivative of the latter expression with respect to $\alpha$ we observe that its maximum is achieved if $\alpha^2 = \frac{(t - 1)^2}{4 s_t}$. This implies that $\mathbb E[(g^t_{ui} - \tilde g^t_{ui}) A x^{t - 1}] \le \frac{2 (t - 1)}{\sqrt{s_t}}e^{-1/2} \le \frac{2t}{\sqrt{s_t}}$. 

%\le (\hat x^t - \hat x^{t - 1})^T A q^{t - 1},
%Using the expression for the components of $\hat x^t_u - \hat x^{t - 1}_u$ this implies that $(\hat x^t - \hat x^{t - 1})^T A x^{t - 1} \le 0$.
\end{proof}

By Lemma~\ref{lem:step-bound-aux} it suffices to bound $\mathbb E\left[ (\hat x^t - \hat x^{t - 1})^T A \hat x^{t - 1}\right]$. We have:
\begin{align*}
(\hat x^t - \hat x^{t - 1})^T A \hat x^{t - 1}  &= \frac{n}{t - 1} (\hat x^t - \hat x^{t - 1})^T A  x^{t - 1} + (\hat x^t - \hat x^{t - 1})^T A q^{t - 1}.  
\end{align*}

We bound the first term using Proposition~\ref{prop:sampling-concentration}.
If $u \notin S^{t - 1}$, but $u \in S^t$ then $\hat x^t_u - \hat x^{t - 1}_u = g^t_u - \hat x^{t - 1}_u$. Hence, $\mathbb E[(\hat x^t_{ui} - \hat x^{t - 1}_{ui}) A x^{t - 1}] = \frac{2 t}{\sqrt{s_t}}$.
If $u \notin S^t$ then $\hat x^t_u - \hat x^{t - 1}_u = \frac{1}{t}\left(g^t_u - \hat x^{t - 1}_u\right)$ and $\mathbb E[(\hat x^t_{ui} - \hat x^{t - 1}_{ui}) A x^{t - 1}] = \frac{2}{\sqrt{s_t}}$.  Since the total number of such vertices is $n - t$ we have:
\begin{align*}
\mathbb E \left[(\hat x^t - \hat x^{t - 1})^T A  x^{t - 1}\right] = \frac{4t}{\sqrt{s_t}} + \frac{4(n - t)}{\sqrt{s_t}} = \frac{4 n}{\sqrt{s_t}}.
\end{align*}
Thus, the first term is bounded by $\frac{4 n^2}{(t - 1) \sqrt{s_t}}$ as desired.

Now we bound the second term. Taking expectation we have:
\begin{align*}
 &\mathbb E \left[(\hat x^t - \hat x^{t - 1})^T A q^{t - 1} \right] = \mathbb E_{S^{t - 1}} \left[\mathbb E_{r_t} \left[(\hat x^t - \hat x^{t - 1})^T A q^{t - 1} | S^{t - 1}\right] \right]  \\
 &= \mathbb E_{S^{t - 1}} \left[ \sum_{v} \mathbb E_{r_t} \left[(\hat x^t_{v} - \hat x^{t - 1}_{v}) A_{v} q^{t - 1} | S^{t - 1}\right] \right] = \mathbb E_{S^{t - 1}} \left[ \sum_{v,i} A_{v,i} q^{t - 1} \mathbb E_{r_t} \left[(\hat x^t_{v,i} - \hat x^{t - 1}_{v,i}) | S^{t - 1}\right] \right]\\
& \ge \mathbb E_{S^{t - 1}} \left[ -  \sum_{v} \left|A_{v} q^{t - 1}\right|_1  \mathbb E_{r_t} \left[\left|\hat x^t_{v} - \hat x^{t - 1}_{v} \right|_1 | S^{t - 1}\right]  \right] \ge \mathbb E_{S^{t - 1}} \left[ -  \frac{2 n}{t (n - t + 1)}   \sum_{v} \left|A_{v} q^{t - 1}\right|_1   \right] =  -  \frac{2 n}{t (n - t + 1)}   \mathbb E \left[ \left|A q^{t - 1}\right|_1   \right]
\end{align*}
where in the second inequality we bound
 $ \mathbb E_{r_t} \left[\left|\hat x^t_{v} - \hat x^{t - 1}_{v} \right| |_1 S^{t - 1} \right]$ as follows:
\begin{align*}
&\mathbb E_{r_t} \left[\left|\hat x^t_{v} - \hat x^{t - 1}_{v} \right|_1 | S^{t - 1} \right] 
= \Pr[r_t = v | S^{t - 1}] \cdot |g^t_v - \hat x^{t - 1}_v|_1 + \Pr[r_t \neq v | S^{t - 1}]  \cdot |g^t_v - \hat x^{t - 1}_v|_1 / t \\
& = \cdot |g^t_v - \hat x^{t - 1}_v |_1 \left(\frac{1}{n - t + 1} + \frac{1}{t} \left(1 - \frac{1}{n - t + 1}\right)\right) = \frac{n}{t (n - t + 1)} |g^t_v - \hat x^{t - 1}_v|_1 \le \frac{2n}{t (n - t + 1)}. \qedhere
\end{align*}
\end{proof}

The following lemma is proved in Appendix~\ref{app:martingales}. \todo{Moved the proof to the appendix.}
\begin{lemma}\label{lem:additive-error}
For all $t \ge t_0$  it holds that $\mathbb E \left[\left|A q^t\right|_1\right] = O\left(n (n - t) \left(\frac{1}{\sqrt{t}} + \frac{1}{\epsilon t}\right)\right) = O\left(\frac{n (n - t)}{\sqrt{t}}\right)$
\end{lemma}

\begin{lemma} \label{lem:key-lemma}
For all $t \ge t_0$ it holds that $\mathbb E \left[(\hat x^t)^T A \hat x^t\right] - \mathbb E\left[(\hat x^{t_0})^T A \hat x^{t_0}\right] \le O(\epsilon n^2).$
\end{lemma}

\begin{proof}
The proof follows from Lemma~\ref{lem:step-bound} and Lemma~\ref{lem:additive-error}.
Using Lemma~\ref{lem:step-bound} applied for $\tau$ from $t_0$ to t:
\begin{align*}
&\mathbb E \left[(\hat x^t)^T A \hat x^t - (\hat x^{t_0})^T A \hat x^{t_0}\right] \le \sum_{\tau = t_0 + 1}^{t} \frac{4 n^2}{\tau^2} + \sum_{\tau = t_0 + 1}^{t} \frac{4 n^2}{(\tau - 1) \sqrt{s_\tau}} + \sum_{\tau = t_0 + 1}^t\frac{2n}{\tau (n - \tau + 1)} \mathbb E \left[ \left|A q^{\tau - 1}\right|_1\right] \\ 
&\le  O(\epsilon^2 n^2) + \sum_{\tau = t_0 + 1}^{t} \frac{4 n^2}{(\tau - 1) \sqrt{s_\tau}} + \sum_{\tau = t_0 + 1}^t O\left(\frac{n^2}{\tau^{3/2}}\right) \le O(\epsilon^2 n^2) + O(\epsilon n^2) + \sum_{\tau = t_0 + 1}^t O\left(\frac{n^2}{\tau^{3/2}}\right) \\
&\le O(\epsilon n^2) + O\left(\frac{n^2} { \sqrt{t_0}}\right)  \le O(\epsilon n^2), 
\end{align*}
where the second inequality is by Lemma~\ref{lem:additive-error}, the third is because $\sum_{\tau = t_0 + 1}^t \frac{4 n^2}{(\tau - 1)\sqrt{s_\tau}} \le \frac{\epsilon}{n^{1/3}}\sum_{\tau = t_0}^n \tau^{- 2/3} = O(\epsilon n)$ \todo{Time-dependent sampling rate in the analysis of the approxiamtion guarantee.}  and the fourth is by $\sum_{\tau = t_0 + 1} ^n \frac{1}{\tau^{3/2}} \le \int_{\tau = t_0}^{\infty}\frac{1}{\tau^{3/2}} d \tau = O\left(\frac{1}{\sqrt{t_0}}\right)$.
\end{proof}

%\begin{proof}
%The proof follows from Lemma~\ref{lem:step-bound} and Lemma~\ref{lem:additive-error}.
%Using Lemma~\ref{lem:step-bound} applied for $\tau$ from $t_0$ to t:
%\begin{align*}
%&\mathbb E \left[(\hat x^t)^T A \hat x^t - (\hat x^{t_0})^T A \hat x^{t_0}\right] \le \sum_{\tau = t_0 + 1}^{t} \frac{4 n^2}{\tau^2} + \sum_{\tau = t_0 + 1}^{t} \frac{4 n^2}{(t - 1) \sqrt{s_\tau}} + \sum_{\tau = t_0 + 1}^t\frac{2n}{\tau (n - \tau + 1)} \mathbb E \left[ \left|A q^{\tau - 1}\right|_1\right] \\ 
%&\le  O(\epsilon^2 n^2) + \sum_{\tau = t_0 + 1}^{t} \frac{4 n^2}{(t - 1) \sqrt{s_\tau}} + \sum_{\tau = t_0 + 1}^t O\left(\frac{n^2}{\tau^{3/2}}\right) \le O(\epsilon^2 n^2) + O(\epsilon n^2) + \sum_{\tau = t_0 + 1}^t O\left(\frac{n^2}{\tau^{3/2}}\right) \\
%&\le O(\epsilon n^2) + O\left(\frac{n^2} { \sqrt{t_0}}\right)  \le O(\epsilon n^2), 
%\end{align*}
%where the second inequality is by Lemma~\ref{lem:additive-error}, the third is because $\sum_{t_0 + 1}^t \frac{1}{t - 1} = O(\log \epsilon^2 n)$ and $s = \log^2 n / \epsilon^2 $ and the fourth is by $\sum_{\tau = t_0 + 1} ^n \frac{1}{\tau^{3/2}} \le \int_{\tau = t_0}^{\infty}\frac{1}{\tau^{3/2}} d \tau = O\left(\frac{1}{\sqrt{t_0}}\right)$.
%\end{proof}

Finally, we are ready to prove Theorem~\ref{thm:max-cut-greedy-sampling}.

\begin{proofof}{Theorem~\ref{thm:max-cut-greedy-sampling}}
Follows from Lemma~\ref{lem:key-lemma} applied to $t = n$ since $\hat x^{t_0} = x^*$ is the optimal cut.
\end{proofof}

\section{Fast algorithm for $r$-CSPs}\label{sec:main}
In order to generalize our algorithms to general $r$-CSPs we need to introduce a notion of a critical constraint.
\begin{definition}[Critical $r$-tuples and constraints]\label{def:critical-constraint}
For a given partial assignment $S$ and a variable $v_t$ an $r$-tuple $(i_1, \dots, i_r)$ is critical, if $i_1 = t$ and $i_2, \dots, i_r \in S$.
A constraint is \textit{critical} if it is defined on a set of variables, whose indices form a critical $r$-tuple, and is not satisfied by $S$.
\end{definition} \todo{Changed the definition.}

\begin{algorithm}
\label{alg:ptas-rcsp}
\caption{Fast Greedy PTAS with subsampling.}\label{alg:greedy-ptas-sampling-fast}
\SetKwInOut{Input}{input}\SetKwInOut{Output}{output}
\Input{A $k$-ary $r$-CSP instance over $n$ variables, parameter $\eps$.}
%\Output{Accept/Reject}
\DontPrintSemicolon
\BlankLine
\nl Pick a sample $S_1$ of $t_1 = O(\log^2 k / \epsilon^4)$ variables uniformly at random without replacement \;
\nl Pick a sample $S_0 \subseteq S_1$ of $t_0 = 1/\epsilon^2$ variables uniformly at random without replacement\;
\nl \label{ln:first-loop} \For {each of the $k^{t_0}$ possible assignments of values to variables in $S_0$}
{ \nl $S^{t_0}_0 = S_0, t = t_0 + 1, s_t =O\left(\frac{n^{2/3}k^4}{t^{2/3}\epsilon^2}\right)$. \;
\nl \label{ln:greedy-step}  \For {each variable $v \in S_1 \setminus S_0$ in random order} 
{ \nl Pick a sample $V^t$ of $s_t$ critical $r$-tuples uniformly without replacement from the set of all critical $r$-tuples for $v$ and assignment $S^{t - 1}$. \;
\nl Assign variable $v$ the value maximizing the number of satisfied critical constraints in $V^t$. \;
\nl $S^t_0 = S^{t - 1}_0 \cup \{v\}, t = t + 1$ }}
\nl Assign values to variables in $S_1$ according to the best assignment $w$ found  over all iterations in line~\ref{ln:first-loop}\;
\nl \label{ln:greedy-pass} \For {each variable $v \in V \setminus S_1$ in random order}
{
	 \nl $S^{t_1}_1 = S_1$, $t = t_1 + 1, s_t = O\left(\frac{n^{2/3}k^4}{t^{2/3}\epsilon^2}\right)$. \;
\nl Pick a sample $V^t$ of $s_t$ critical $r$-tuples uniformly without replacement from the set of all critical $r$-tuples for $v$ and assignment $S^{t - 1}$.\;
	\nl Assign variable $v$ the value maximizing the number of satisfied critical constraints in $V^t$. \;
	\nl $S^t_1 = S^{t - 1}_1 \cup \{v\}, t = t + 1$ \;
}
\nl Output the assignment constructed in the loop on line~\ref{ln:greedy-pass}
\end{algorithm}

\todo{Changed the description of the sampling process in the algorithm as per Comment 6 and introduced time-dependent sample size.}

The assignment of values to the variables is represented by a vector $x \in \{0,1\}^{nk}$, where $x_{ui} = 1$ iff the variable $u$ is assigned value $i$.
The objective function can be written as a multilinear function:
\begin{align*}
A(x^{(1)}, \dots, x^{(r)}) = \sum_{\substack{1 \le u_1, \dots, u_r \le n  \\ 1 \le i_1, \dots, i_r \le k}} A_{u_1, i_1, \dots, u_r, i_r} x^{(1)}_{u_1, i_1} \dots, x^{(r)}_{u_r, i_r},
\end{align*}
where $A$ is an $nk$-dimensional array symmetric under permutation of the $r$ indices $(u_j, i_j)$.
The variables $x^t_{ui}$ and $\hat x^t_{ui}$, are defined as in Section~\ref{sec:max-cut-warmup}, except that we now call them assignments instead of cuts as before. 
Random variable $r_t$ corresponds to the variable chosen at random at step $t$. Random variable $g^t$ denote the optimum greedy choice of the assignment for this variable with respect to all its critical constraints in $S^{t - 1}$.  
Random variable $\tilde g^t$ denotes the optimum such greedy choice but only with respect to constraints in $V^t$.
We also define an array $A^t$ sampled from $A$ at iteration $t$ of the algorithm as follows:
\begin{align*}
A^t_{u_1, i_1, u_2, i_2, \dots, u_r, i_r} = 
\begin{cases}
A_{u_1, i_1, u_2, i_2, \dots, u_r, i_r}, \text{ if} (u_1 = r_t, u_2, u_3, \dots, u_r) \in V_t, \text{i.e. the constraint is critical.} \\
0, \text{ otherwise.}
\end{cases}
\end{align*}
We also introduce notation an $nk$-dimensional vector $A(\cdot, x, \dots, x)$ with components defined as $A(\cdot, x, \dots, x)_{ui} = A(e_{ui}, x, \dots, x)$.

Recall that $q^t = \hat x^t - \frac{n}{t} x^t$. We use notation $S^{t}_1$ to denote the set of variables, which  are assigned values after iteration $t$ of the loop on line~\ref{ln:greedy-pass} of Algorithm~\ref{alg:ptas-rcsp}. 

\begin{lemma}\label{lem:step-bound-rcsp} (Analog of Lemma~\ref{lem:step-bound})
For every $t \ge t_1$ it holds that: \todo{Corrected statement as per Comment 2.}
\begin{align*}
\mathbb E\left[A(\hat x^t, \dots, \hat x^t) - A(\hat x^{t - 1}, \dots, \hat x^{t - 1})\right] \le \frac{2^{r + 2} n^r}{t^r} + \frac{2 k^2 n^r }{t \sqrt{s_t}} + \frac{2 n}{t (n - t + 1)} \mathbb E\left[|A(\cdot, q^{t - 1},\dots, q^{t - 1})|_1\right]
\end{align*}
\end{lemma}
\begin{proof}
First we prove the following proposition, which generalizes Proposition~\ref{prop:sampling-concentration}.
\begin{proposition}\label{prop:sampling-concentration-rcsp}
For every $u \notin S^{t-1}_1$ it holds that $\mathbb E\left[A(g^t_{ui} - \hat x^{t - 1}_{ui}, x^{t - 1}, \dots, x^{t - 1})\right] \le \frac{2 k t^{r - 1}}{\sqrt{s_t}}$
\end{proposition}
\begin{proof}
The proof generalizes the proof of Proposition~\ref{prop:sampling-concentration}.
There are two differences: we now work with an alphabet of size $k > 2$ and have an instance of an $r$-CSP instead of an instance of \textsc{Max-Cut}.
The size of the alphabet can be taken care of by a union bound, which introduces an extra factor of $k$ in the result. To handle larger arity of the CSP recall the definition of a \textit{critical constraint} (Definition~\ref{def:critical-constraint}).
It captures the intuition that the hardest type of $r$-CSP predicates for us are predicates $r$-LIN, i.e. parities on at most $r$ variables which can't be satisfied until all variables are assigned values (a special case $2$-LIN corresponds to \textsc{Max-Cut}).
The key observation is that at time $t$ there can be at most $t^{r - 1}$ critical constraints, which have $r_t$ as the unassigned variable since there are at most that many cricical $r$-tuples. Using these two observations the proof follows the lines of the proof of Proposition~\ref{prop:sampling-concentration} with $t$ replaced by $t^{r - 1}$, which is used as a new bound on the size of the  set we sample from.
\end{proof}
The rest of the proof of  Lemma~\ref{lem:step-bound-rcsp} is given in Appendix~\ref{app:proof-step-bound-rcsp}. \todo{Proof added.}
\end{proof}

\begin{lemma}\label{lem:additive-error-rcsp}(Analog of Lemma~\ref{lem:additive-error})
For every $t$ and $\sigma = O\left(\frac{n^r}{\sqrt{t}} \sqrt{\frac{n - t}{n}}\right)$ it holds that:
\begin{align*}
\mathbb E\left[|A(\cdot, q^t, \dots, q^t)|_1\right] = O(\sigma).
\end{align*}
\end{lemma}
\begin{proof}

\end{proof}

\begin{lemma}\label{lem:key-lemma-rcsp}
For every $t \ge t_0$ it holds that:
\begin{align*}
\mathbb E \left[A(x^t, \dots, x^t) - A(\hat x^{t_0}, \dots, \hat x^{t_0})\right] \le O(\epsilon n^r).
\end{align*}
\end{lemma}
\begin{proof}
By Lemma~\ref{lem:step-bound-rcsp} and Lemma~\ref{lem:additive-error-rcsp} we have:
\begin{align*}
&\mathbb E \left[A(x^t, \dots, x^t) - A(\hat x^{t_0}, \dots, \hat x^{t_0})\right] \le O\left(\sum_{\tau = t_0}^t \frac{2^{r + 2} n^r}{\tau^r} + \sum_{\tau = t_0}^t \frac{2 k^2 n^r}{\tau \sqrt{s_\tau}} + \sum_{\tau = t_0}^t \frac{2n}{\tau (n - \tau + 1)} \frac{n^r }{\sqrt{\tau}} \sqrt{\frac{n - \tau}{n}}\right) \\
&= O \left(\epsilon n^r + 2 k^2 n^r\int_{t_0}^{t} \frac{1}{\tau \sqrt{s_\tau}}d \tau + n^r \sum_{\tau = t_0}^{n/2} \frac{1}{\tau^{3/2}}  + n^{r + 1/2}\sum_{n / 2}^n \frac{\sqrt{n - \tau}}{\tau^{3/2} (n - \tau + 1)}\right) \\
& = O \left(\epsilon n^r + \epsilon n^r  + \frac{n^r}{\sqrt{t_0}} + n^{r - 1} \int_{1}^{n/2} \frac{1}{\sqrt{z}} dz\right) = O(\epsilon n^r),
\end{align*}
where the first equality follows by direct calculations, the third inequality uses the definition $s_\tau = O\left(\frac{n^{2/3} k^4}{\tau^{2/3} \epsilon^2}\right)$ and the last two equalities follow by direct calculations \todo{Modified calculations}.
%If $r = 2$ then for the integral in the second term of the sum above we have $\int_{1 / \epsilon^2}^n \frac{1}{\tau} d \tau = O(\log n)$ and since $s = O\left(\frac{k^2 \log^2 n}{\epsilon^2}\right)$ the bound follows.
%If $r > 2$ then the integral is $\int_{1 / \epsilon^2}^n \tau^{r - 3} d \tau = O(n^{r - 2})$ and since $s = O(k^2 / \epsilon^2)$ the bound follows.
\end{proof}

We use notation $x^t_{(y)}$, $\hat x^t_{(y)}$ and $q^t_{(y)}$ to denote the corresponding variables at time $t$ in the greedy process starting from the assignment $y$ instead of the optimum assignment $x^*$ as before. Let $Y$ be the set of all $k^{t^0} = 2^{O(\log k / \epsilon^2)}$ possible assignments of variables the set $S_0$ that the Algorithm~\ref{alg:ptas-rcsp} tries in the line~\ref{ln:first-loop}.

\newcommand{\sigmaval}[1]{\frac{n^{r - 1}}{\sqrt{#1}}\sqrt{\frac{n - #1}{n}}}
\newcommand{\myb}[1]{A\left(\cdot, #1, \dots, #1\right)}
\newcommand{\myz}[1]{A\left(#1, #1, \dots, #1\right)}

The proof of the following lemma is given in Appendix~\ref{app:union-bound}.

\begin{lemma}\label{lem:union-bound}
For every $t \ge t_0$ and $\sigma = O\left( \sigmaval{t} \right)$ it holds that:
\begin{align*}
\mathbb E\left[\max_{y \in Y} \left|A(\cdot, q^t_{(y)}, \dots, q^t_{(y)})\right|_1\right] &= O\left(\frac{\sigma \log k}{\epsilon}\right),\\
\mathbb E\left[\max_{y \in Y} \left|A(q^t_{(y)}, q^t_{(y)}, \dots, q^t_{(y)})\right|\right] &= O\left(\frac{\sigma \log k}{\epsilon}\right)
\end{align*}
\end{lemma}

Now we are ready to complete the analysis of Algorithm~\ref{alg:ptas-rcsp}, giving the proof of Theorem~\ref{thm:rcsp}

\begin{proofof}{Theorem~\ref{thm:rcsp}}
Let $w$ denote the best assignment found over all iterations of the loop in line~\ref{ln:first-loop}, which is used as a seed assignment to the vertices in $S_1$ for the second loop in line~\ref{ln:greedy-pass}. We have:
\begin{align}
&\mathbb E\left[\myz{\hat x^n_{w}}\right] = \nonumber \\ &\mathbb E\left[\myz{\hat x^{t_1}_{(w)}}\right] + \sum_{t_1 \le t \le n} \mathbb E\left[\myz{\hat x^t_{(w)}} - \myz{\hat x^{t - 1}_{(w)}}\right] \label{eq:main-bound}
\end{align}
By Lemma~\ref{lem:step-bound-rcsp} for each term in the sum in~(\ref{eq:main-bound}) we have:
\begin{align*}
\mathbb E \left[\myz{\hat x^t_{(w)}} - \myz{\hat x^{t - 1}_{(w)}}\right] \le 
\frac{2^{r + 2} n^r}{t^r} + \frac{4 k n^2 t^{r - 3}}{\sqrt{s}} + \frac{2 n}{t (n - t + 1)} \mathbb E\left[|\myb{q^{t - 1}_{(w)}}|_1\right]
\end{align*}
We have that $\left|\myb{q^{t - 1}_{(w)}}\right|_1 \le \max_{y \in Y} \left|\myb{q^{t - 1}_{(y)}}\right|_1$. Thus, by Lemma~\ref{lem:union-bound} the last term in the expression above can be bounded by $O\left(\frac{n}{t (n - t + 1)} \sigmaval{t} \frac{\log k}{\epsilon}\right).$ 
Bounding the sum by integral as in the proof of Lemma~\ref{lem:key-lemma-rcsp} we get that the sum in the second term of~\ref{eq:main-bound} is bounded as $O\left(\frac{n^r}{\sqrt{t_1}} \frac{\log k}{\epsilon} + \epsilon n^r\right)$.
This is where we use the fact that $t_1 = \Theta\left(\frac{\log^2 k}{\epsilon^4}\right)$ to conclude that the sum is bounded by $O(\epsilon n^r)$.

For the first term $\mathbb E\left[\myz{\hat x^{t_1}_{(w)}}\right]$ of~(\ref{eq:main-bound}) we have:
\begin{align*}
\mathbb E\left[\myz{\hat x^{t_1}_{(w)}} \right]&\le \mathbb E\left[\myz{\frac{n}{t_1}\hat x^{t_1}_{(w)}}\right] + \mathbb E\left[\left|\myz{q^{t_1}_{(w)}}\right|\right] \\
& \le \mathbb E\left[ \myz{\frac{n}{t_1}\hat x^{t_1}_{(w)}} \right] + \mathbb E \left[\max_{y \in Y} \left|\myz{q^{t_1}_{(y)}}\right|\right]\\
& \le \mathbb E\left[\myz{\frac{n}{t_1} x^{t_1}_{(w)}} \right] + O(\epsilon n^{r - 1})\\
&\le \mathbb E\left[\myz{\frac{n}{t_1} x^{t_1}_{(x^*)}}\right] + O(\epsilon n^{r - 1})\\
&\le \mathbb E\left[\myz{\hat x^{t_1}_{(x^*)}} \right]+ \mathbb E \left[\left|\myz{q^{t_1}_{(x^*)}}\right| \right]+ O(\epsilon n^{r - 1})\\
&\le \mathbb E\left[\myz{\hat x^{t_1}_{(x^*)}} \right]+ \mathbb E \left[ \max_{y \in Y}\left|\myz{q^{t_1}_{(y)}}\right| \right] + O(\epsilon n^{r - 1})\\
& \le \mathbb E\left[\myz{\hat x^{t_1}_{(x^*)}} \right] + O\left(\sigmaval{t_1} \frac{\log k}{\epsilon} \right) + O(\epsilon n^{r - 1})\\
& \le \mathbb E\left[\myz{\hat x^{t_0}_{(x^*)}} \right]  + O\left(\epsilon n^r\right) \\
& = OPT + O(\epsilon n^r),
\end{align*} 
where the first inequality is by the triangle inequality and the definition of $q^{t_1}_{(1)}$, the second inequality is by a union bound, the third inequality is by Lemma~\ref{lem:union-bound}, using the fact that $O\left(\sigmaval{t_1} \frac{\log k}{\epsilon}\right) = O\left(\frac{n^{r - 1}}{\sqrt{t_1}} \frac{\log k}{\epsilon}\right) = O(\epsilon n^{r - 1})$, the fourth inequality follows since $x^* \in Y$, the fifth inequality holds by the triangle inequality and the definition of $q^{t_1}_{x^*}$, the sixth inequality is by a union bound, the seventh inequality is by Lemma~\ref{lem:union-bound}, the eighth inequality is by Lemma~\ref{lem:key-lemma-rcsp} and uses the fact that $t_1 = \Theta(\log^2 k / \epsilon^4)$ and the last equality holds since by definition $\myz{\hat x^{t_0}_{(x^*)}} = OPT$.

Finally, the analysis of the running time follows from computing the overall size of all samples as done for Algorithm~\ref{alg:greedy-ptas-sampling}. 
In order to implement random sampling from the set of all critical $r$-tuples we assume that the instance is given the input is given as the $r$-dimensional matrix $A$. Then we can check whether a randomly sampled critical $r$-tuple corresponds to a critical constraint in constant time.
\todo{Discussion of the running time and implementation of sampling.}
\end{proofof}

\section{Conclusion}\label{sec:conclusions}
In this section we briefly explain how to use our algorithms in a parallel setting, when we have $m$ identical machines available as in modern massive parallel computational models~\cite{FMSSS10, KSV10, GSZ11, BKS13, ANOY14}.
This allows to reduce the total computational time of our algorithms by a factor of $m$.
The loop in the line~\ref{ln:first-loop} is very easy to parallelize since different iterations can be done independently on different machines.
The greedy pass in the line~\ref{ln:greedy-step}, however, is executed only once and because the decisions made in the previous iterations affect the greedy choices in the subsequent steps we have to modify the algorithm in order to make it parallel.
We can do this by partitioning the work into $O(\log_{1 + \epsilon} n)$ supersteps.
Instead of processing one vertex at a time in the loop in the line~\ref{ln:greedy-step} in every superstep we increase the size of the current set of processed vertices by a $(1 + \epsilon)$ multiplicative factor. Consider one such superstep $\tau$ when we increase the set of vertices from $S^{\tau} - 1$ to $S^{\tau}$. In this superstep we process all new vertices in $S^{\tau} \setminus S^{\tau -1}$ independently in random order, dividing the work uniformly between $m$ machines. Crucially, when processing a vertex $r_t \in S^{\tau} \setminus S^{\tau -1}$ we can no longer take the random sample $V_t$ from the entire set of vertices $r_1, \dots, r_{t - 1}$ processed before it. However, we can take a random sample from $S^{\tau - 1}$, which was computed in the previous superstep and differs from $\{r_1, \dots, r_{t - 1}\}$ on at most an $\epsilon$-fraction of points. This suffices for the analysis in Section~\ref{sec:main} to still yield a PTAS.

It remains open whether our results can be reproduced in the streaming model, i.e. whether there exist both sublinear time and space algorithms for dense $r$-CSPs.

\section{Acknowledgment}
We would like to thank an anonymous reviewer for SODA'15 for multiple insightful comments, including a suggestion to use time-dependent sample size.

\bibliographystyle{alpha}
\bibliography{maxcut}

\appendix

\section{Lower bound}\label{sec:lower-bound}

\begin{theorem}
There exist constants $c_1, c_2 > 0$ such that any algorithm that queries less than $c_1 n / \epsilon^2$ edges has to incur additive error at least $c_2 \epsilon n^2$ for \textsc{Max-Cut}.
\end{theorem}
\begin{proof}
Consider the following  random family of graphs. Let $V = V_0 \cup V_1 \cup V_2$, where $|V_0| = |V_1| =  4 n / 9$ and $|V_2| = n / 9$. The subgraph induced by $V_0 \cup V_1$ is a complete bipartite graph $K_{4n/9, 4n/9}$. For every vertex $v \in V_2$ we construct the set of edges adjacent to $v$ by picking a random value $r_v \in \{0,1\}$ uniformly at random and choosing 
drawing edges from $v$ to vertices in $V_{r_{v}}$ independently with probability $1/2 + \epsilon$ each and edges to vertices in $V_{1 - r_{v}}$ independently with probability $1/2$ each.
With high probability this gives us two random sets of neighbors $S_{r_v} \subseteq V_{r_v}$ of size $ 2 n / 9 \pm o(n)$ and $S_{1-r_v} \subseteq V_{1 - r_v}$ of size $(1 + \epsilon) 2n/ 9 \pm o(n)$. In the presentation below we will often omit the $o(n)$ terms since they don't matter for the analysis.

First, we give a high-probability bound on the cost of the optimum solution for such instances.

\begin{proposition}
With high probability the cost of the optimum solution for the class of instances constructed as described above is equal to $18 n^2 / 81 + 2 \epsilon n^2 / 81 + o(n^2)$.
\end{proposition}
\begin{proof}
We use sets $C_0$ and $C_1$ for the two sides of the cut.
First, we show that the optimum solution for this family of instances always places $V_0$ in $C_0$, $V_1$ in $C_1$ and every vertex $v \in V_2$ in $C_{r_v}$ and thus has the size of the cut equal to $OPT = |V_0| |V_1| + |V_2| |S_{r_v}| = 4n / 9 \cdot 4n/ 9 + n/ 9 \cdot (1 + \epsilon) 2n/9 = 18 n^2 / 81 + 2 \epsilon n^2 / 81$. Indeed, we can choose $C_0$ so that the majority of vertices from $V_0$ are placed in $C_0$. We denote this majority as $V^0_0 = V_0 \cap C_0$ and define analogously $V^i_j = V_j \cap C_i$ for $i, j \in \{0,1\}$. 
Note that the total number of cut edges in the subgraph induced by $V_0 \cup V_1$ is $|V^0_0|\cdot |V^1_1| + |V^1_0|\cdot|V^0_1|$, which is maximized if $|V^1_1| = |V_1| = 4n / 9$. Thus, the number of such edges is at most $|V^0_0| 4n / 9$.
The number of cut edges adjacent to the vertices in $V_2$ is at most $|V_2| \cdot \left(|S_{r_v}| + |S_{1 - r_v}|\right) = |V_2| \cdot (4n / 9 + 2 \epsilon n / 9) = 4 n^2 / 81 + 2 \epsilon n^2 / 81$.
Thus, the total number of cut edges is at most $|V^0_0| 4n / 9 + 4 n^2 / 81 + 2 \epsilon n^2 / 81$. This implies that $|V^0_0| \ge (OPT - 4 n^2 / 81 - 2 \epsilon n^2 / 81) / (4n / 9) = 7n / 18 > n / 3$. Now assume that $V^0_1$ is non-empty and consider any vertex in $v \in V^0_1$. The number of edges from $v$ to vertices in $C_0$ is at least $|V^0_0| > n / 3$, while the number of edges from $v$ to vertices in $C_1$ is at most $|V^1_0| + |V_2| < n / 9 + n / 9 = 2n / 9$. By moving $v$ to $C_1$ we can improve the solution, which means that $V^0_1$ is empty. Similarly, $V^1_0$ is empty and hence $V^0_0  = V_0$ and $V^1_1 = V_1$. Given this the optimum solution places every vertex $v \in V_2$ greedily on the side $r_v$, which completes the analysis of the cost of the optimum solution.
\end{proof}

Using Yao's principle it suffices to consider performance of deterministic algorithms under the hard distribution that we have constructed.
\begin{definition}[Bad vertices]
A vertex $v \in V_2$ is \textit{bad} if there exists a cut $(C_0, C_1)$ of $V_0 \cup V_1$ such that the degree of $v$ in on of the parts $V^i_j$ the partition deviates from the expectation by more than $(4n / 9)^{3/4}$.
\end{definition}

\begin{proposition}
With high probability there are at most $n^{3/4}$ vertices.
\end{proposition}
\begin{proof}
We denote $m = 4n / 9$.
Fix a vertex $v \in V_2$.
W.l.o.g $r_v = 0$. For a set $S$ let $n_v(S)$ denote the number of neighbors of $v$ in $S$. We have $\mathbb E[n_v(V^0_0 \cup V^0_1)] =  \mathbb E[n_v(V^0_0)] + \mathbb E[n_v(V^0_1)] = |V^0_0| (1 + \epsilon) 1/2 + |V^0_1| 1/2 = 1/2 (|V^0_0 \cup V^0_1|) + \epsilon / 2 |V^0_0|$.
By a union bound $\Pr[n_v(V^0_0 \cup V^0_1) - \mathbb E[n_v(V^0_0 \cup V^0_1)] > 2m^{3/4}] \le \Pr[n_v(V^0_0) - \mathbb E[n_v(V^0_0)] > m^{3/4}] + \Pr[n_v(V^0_1) - \mathbb E[n_v(V^0_1)] > m^{3/4}]$.
By Hoeffding bound we have $\Pr[n_v(V^0_0)  - \mathbb E[n_v(V^0_0)]> m^{3/4}] \le e ^{- \frac{2 m^{3/2}}{|V^0_0|}} \le e^{- 2 \sqrt{m}}$ and similarly $\Pr[n_v(V^0_1) - \mathbb E[n_v(V^0_1)] > m^{3/4}] \le e^{-2 \sqrt{m}}$. Thus $\Pr[n_v(V^0_0 \cup V^0_1) - \mathbb E[n_v(V^0_0 \cup V^0_1)] > 2m^{3/4}] \le 2 e^{-2 \sqrt{m}}$.
The probability that for a fixed cut there are at least $n^{3/4}$ bad vertices is thus $\sum_{t = n^{3/4}}^{|V_2|} \binom{|V_2|}{t} p^{t} (1 - p)^{|V_2| - t} \le \sum_{t = n^{3/4}}^{|V_2|} \binom{|V_2|}{t} p^{t} \le p^{n^{3/4}} \sum_{t = n^{3/4}}^{|V_2|} \binom{|V_2|}{t} \le p^{n^{3/4}} 2^{|V_2|} \le (2 e^{- 2 \sqrt{4n / 9}})^{n^{3/4}} 2^{n/9}$.
Finally, taking a union bound over all $2^{2m} = 2^{8n / 9}$ possible cuts induced on $V_0 \cup V_1$ we get that the probability that there exists at least one such cut for which at least $n^{3/4}$ vertices are bad is at most $2^{n + n^{3/4}} \cdot e^{- 4/3 n^{5/4}} = 2^{- \Omega(n^{1/4})}$, which is exponentially small in $n$.
\end{proof}

The previous lemma essentially allows us to assume that there are no bad vertices, becuase their total contribution to the cut is at most $o(n^2)$, which is negligible. Now we define \textit{biased vertices}, which intuitively are vertices in $V_2$ which are easy for an algorithm to detect since they have a much larger than expected bias towards one of the sides of the planted solution on the vertices $V_0 \cup V_1$. Most importantly, bad vertices are defined with respect to an arbitrary cut on $V_0 \cup V_1$ since we don't know which of these cuts was selected by the algorithm.

\begin{definition}[Biased vertices]
We say that a vertex $v \in V_2$ has a \emph{bias} towards the side $C_0$ of the cut, if $n_v(V^0_0 \cup V^0_1) - n_v(V^1_0 \cup V^1_1) > \epsilon n$. Vertices biased towards the side $C_1$ of the cut are defined analogously.
\end{definition}

\begin{proposition}
There no vertices in $V_2$, which a bias towards one of the sides of the cut induced by the algorithm on $V_0 \cup V_1$.
\end{proposition}
\begin{proof}
Suppose that there exists a vertex $v$ with $r_v = 0$ (the case $r_v = 1$ is symmetric) such that $n_v(V^0_0 \cup V^0_1) - n_v(V^1_0 \cup V^1_1) > \epsilon n$. Because we assume that there are no bad vertices this implies that $\mathbb E[n_v(V^0_0 \cup V^0_1)] - \mathbb E[n_v(V^1_0 \cup V^1_1)] > \epsilon n - 2 m^{3/4} > 3 \epsilon n / 4$, where the last inequality holds for large enough $n$.
Expanding the expectations we have $(1 + \epsilon) 1/ 2 |V^0_0| + 1/2|V^0_1| - (1 + \epsilon) 1/ 2 |V^1_0| - 1/2 |V^1_1| > 3 \epsilon n / 4$ or equivalently $|V^0_0| + |V^0_1| - |V^1_0| - |V^1_1| >  (3 \epsilon n / 2 - \epsilon |V^0_0| + \epsilon |V^1_0|) > 3 \epsilon n / 2 - 4 \epsilon n / 9 > \epsilon n$. This implies that all vertices with $r_v = 0$ are optimally placed on the side $C_1$ of the cut.
Similarly for vertices with $r_v = 1$ we have that $\mathbb E[n_v(V^0_0 \cup V^0_1)] - \mathbb E[n_v(V^1_0 \cup V^1_1)] \ge 1/2 |V^0_0| + (1 + \epsilon) 1/2|V^0_1| - 1/2 |V^1_0| - (1 + \epsilon) 1/2 |V^1_1| - 2 m^{3/4} \ge 1/2(|V^0_0| + |V^0_1| - |V^1_0| - |V^1_1|) - \epsilon / 2 |V^1_1| - 2 m^{3/4}\ge \epsilon n / 2 - 2 \epsilon n / 9  - 2 m^{3/4}> 0$ and thus all such vertices are also optimally placed on the $C_1$ side of the cut.

Consider now the number of cut edges in a solution, which is given as $|V^0_0| |V^1_1| + |V^1_0| |V^0_1| +|E_2|$, where by $E_2$ we denote the set of cut edges adjacent to vertices in $E_2$. 
With high probability the number of vertices $v \in V_2$ such that $r_v = 0$ is at most $n / 18 + n^{3/4}$ and the number of vertices such that $r_v = 1$ is also at most $n / 18 + n^{3/4}$.
Because there are no bad vertices $|E_2| \le  (n / 18 + n^{3/4}) \left((1 + \epsilon)1/2|V^0_0| + 1/2|V^0_1| + m^{3/4}\right) + (n / 18 + n^{3/4}) \left(1/2|V^0_0| + (1 + \epsilon)1/2|V^0_1| +  m^{3/4}\right)  = (1 + \epsilon/ 2) n / 18 (|V^0_0| + |V^0_1|) + o(n^2)$.
In the optimum solution vertices in $V_2$ contribute (roughly) $2n^2 / 9 + 2 \epsilon n^2 / 9$ edges.
Thus, the overall gain from vertices in $V_2$ compared to the optimum solution is at most
$(1 + \epsilon / 2) n / 18 (|V^0_0| + |V^0_1|)  - 2 n^2 / 9 - 2 \epsilon n^2 / 9 + o(n^2)$ edges.
This can be written as $n / 9 ((1/2 + \epsilon/4) (|V^0_0| + |V^0_1|) - 2 n / 9 - 2 \epsilon n / 9) + o(n^2)$.
On the other hand, we have $m^2 - |V^0_0||V^1_1| - |V^1_0||V^0_1| = |V^0_0||V^0_1| + |V^1_0||V^1_1| \ge (|V^0_0| + |V^0_1| - |V^1_0| - |V^1_1|) / 2 \cdot |V_0|  = 2 n / 9 (|V^0_0| + |V^0_1| - |V^1_0| - |V^1_1|)$. Putting this together, the overall gain compared to the optimum solution is at most:
\begin{align*}
& n/9 (( (1/2 + \epsilon / 4) (|V^0_0| + |V^0_1|) - 2n / 9 - 2 \epsilon n / 9) - 2 (|V^0_0| + |V^0_1| - |V^1_0| - |V^1_1|) )+ o(n^2) \\
& = n/9 (( (1/2 + \epsilon / 4) (4 n / 9 + (|V^0_0| + |V^0_1| - |V^1_0| - |V^1_1|) / 2 ) - 2n / 9 - 2 \epsilon n / 9) - 2 (|V^0_0| + |V^0_1| - |V^1_0| - |V^1_1|) )+ o(n^2) \\
& = n / 9 ((\epsilon / 8 - 7 /4) (|V^0_0| + |V^0_1| - |V^1_0| - |V^1_1|) - \epsilon n /9) + o(n^2) < 0.
\end{align*}
\end{proof}

Using the fact that there are no biased vertices the rest of the proof follows.
If there are at least $n / 100$ vertices with bias at most $\epsilon n / 100$ then since the bias of these vertices is so small each of them loses $\Omega(\epsilon n)$ cut edges compared to the optimum solution which places these vertices on the side $1 - r_v$ and thus achieves $2n /9 + 2 \epsilon n /9$ cut edges for each such vertex compared to at most $2n/9 + \epsilon n /9 + \epsilon n / 100$ edges for each vertex with a small bias.
Otherwise there are at least $n / 9 - n / 100$ vertices with bias at least $\epsilon n / 100$ and at most $\epsilon n$ (recall that there are no \textit{biased vertices}). Thus, the algorithm has to achieve a non-trivial advantage over probability $1/2$ in guessing the value of $r_v$ for vertices in $V_2$ if it places these vertices on the correct side. By a Chernoff-type sampling lower bound of~\cite{CEG95} this implies that for a constant fraction of vertices in $V_2$ it has to sample at least $\Omega(1 / \epsilon^2)$ edges, which gives the $\Omega(n / \epsilon^2)$ lower bound, completing the proof.
\end{proof}

\section{Omitted proofs}
The key difference between our work and~\cite{MS08, S12} is that we use sampling in the greedy step of our algorithms. However, some of their lemmas aren't affected by this difference.
For completeness we present proofs for the lemmas that we used black-box from~\cite{S12}, which contains an extended version of~\cite{MS08}. 

\subsection{Proof of Lemma~\ref{lem:additive-error}}\label{app:martingales}

First, we prove the following two lemmas:

\begin{lemma} \label{lem:martingale-aux}
$\frac{t}{n - t}q^t_{v i}$ is a martingale.
\end{lemma}

\begin{lemma}\label{lem:martingale}
$\frac{t}{n - t}A_{vk} q^t$ is a martingale with step size bounded by $\frac{4n}{n - t}$.
\end{lemma}

\begin{proofof}{Lemma~\ref{lem:martingale-aux}}
We consider two cases:

\textbf{Case 1:} $v \in S^{t - 1}$. Then we have $\hat x^{t - 1}_{vi} = x^{t - 1}_{vi} = \hat x^{t}_{vi} = x^{t}_{vi}$.
Let's denote this common value as $x$.
Then we have $\frac{t}{n - t} q^t_{vi} = \frac{t}{n - t} \left(x - \frac{n}{t} x\right) = -x$.
Also, $\frac{t - 1}{n - t + 1} q^{t - 1}_{vi} = \frac{t - 1}{n - t + 1} \left(x - \frac{n}{t - 1} x\right) = - x$, as desired.

\textbf{Case 2:} $v \notin S^{t - 1}$. We have $\hat x^{t - 1}_{vi} = \frac{1}{t - 1} \sum_{j = 1}^{t - 1} g^j_{vi}$ and $x^{t - 1}_{vi} = 0$.
This gives $\frac{t - 1}{n - t + 1}q^{t - 1}_{vi} = \frac{1}{n - t + 1} \sum_{j = 1}^{t - 1} g^j_{vi}$.
On the other hand we have $x^t_{vi} = g^t_{vi}$ with probability $\frac{1}{n - t + 1}$ and $x^t_{vi} = 0$ with probability $1 - \frac{1}{n - t + 1}$.
Thus, $\hat x^t_{vi} = g^t_{vi}$ with probability $\frac{1}{n - t + 1}$ and $\hat x^t_{vi} = \frac{1}{t}\sum_{j=1}^{t} g^j_{vi}$ with probability $1 - \frac{1}{n - t + 1}$. This gives us:
\begin{align*}
\mathbb E \left[\frac{t}{n - t}q^t_{vi}\right] = \frac{t}{n - t} \left(\frac{1}{n - t + 1} \left(g^t_{vi} - \frac{n}{t} g^t_{vi}\right)  + \left(1 - \frac{1}{n - t + 1} \right)\frac{1}{t} \sum_{j = 1}^t g^j_{vi}\right)
= \frac{1}{n - t + 1} \sum_{j = 1}^{t - 1} g^j_{vi} = \frac{t - 1}{n - t + 1}q^{t - 1}_{vi}. \qedhere
\end{align*}
\end{proofof}

\begin{proofof}{Lemma~\ref{lem:martingale}}
The martingale property follows from Lemma~\ref{lem:martingale-aux}.
The components $q^t_u$ fall into three cases:

\textbf{Case 1:} $u \in S^{t - 1}$. In this case $q^t_u = q^{t - 1}_u$ so $A_{vk, ui} (q^t_{ui} - q^{t - 1}_{ui}) = 0$.

\textbf{Case 2:} $u  = r^t$. In this case $x^{t}_{ui}  = \hat x^t_{ui}= g^t_{ui}$, $x^{t - 1}_{ui} = 0$ and $\hat x^{t - 1}_{ui} = \frac{1}{t - 1} \sum_{j = 1}^{t - 1} g^j_{ui}$.
This gives us:
\begin{align*}
&\frac{t}{n - t} A_{{vk, ui}} q^t_{ui} - \frac{t - 1}{n - t + 1} A_{vk, ui} q^{t - 1}_{ui} = \frac{t}{n - t} g^t_{ui} \left(1 - \frac{n}{t}\right) - \frac{t - 1}{n - t + 1} \frac{1} {t - 1} \sum_{j= 1}^{t - 1} g^j_{ui}.  \ge - 1 - \frac{t - 1}{n - t + 1} = - \frac{n}{n - t + 1}
\end{align*}
This gives the overall contribution of at most $\frac{2n}{n - t + 1}$.

\textbf{Case 3:} $u \notin S^t$.
In this case we have $x^t_{ui} = x^{t - 1}_{ui} = 0$, $\hat x^t_{ui} = \frac{1}{t}\sum_{j = 1}^t g^j_{ui}$ and $\hat x^{t - 1}_{ui} = \frac{1}{t - 1} \sum_{j = 1}^{t - 1} g^j_{ui}$.
This gives us:
\begin{align*}
&\frac{t}{n - t} A_{{vk, ui}} q^t_{ui} - \frac{t - 1}{n - t + 1} A_{vk, ui} q^{t - 1}_{ui}   \le \frac{t}{n - t}  \frac{1}{t} \sum_{j = 1}^t g^j_{ui} - \frac{t - 1}{n - t + 1}  \frac{1}{t - 1}\sum_{j = 1}^{t - 1} g^j_{ui} \\
&= \frac{g^t_{ui}}{n - t} + \frac{1}{(n - t)(n - t + 1)} \sum_{j = 1}^{t - 1} g^j_{ui} \le \frac{n}{(n - t)(n - t + 1)}.
\end{align*}
The number of vertices $u \notin S^t$ is $n - t$, hence their overall contribution is at most $\frac{2n}{n - t + 1}$.

Putting all together, the step size is bounded by $\frac{4n}{n - t + 1} \le \frac{4 n}{n - t}$.
\end{proofof}

We will use the following version of the Azuma-Hoeffding inequality:

\begin{theorem}[Azuma-Hoeffding]
Let $X_0, X_1, \dots, X_t$ be a martingale such that $|X_k - X_{k - 1}| \le c_k$ for all $k$.
Then for all $\lambda > 0$ it holds that:
\begin{align*}
\Pr[|X_t - X_0| \ge \lambda] \le 2 e^{- \frac{\lambda^2}{2 \sum_{k = 1}^t c^2_k}}
\end{align*}
\end{theorem}

Now we are ready to complete the proof of Lemma~\ref{lem:additive-error}.

\begin{proofof}{Lemma~\ref{lem:additive-error}}
From Azuma-Hoeffding, using a union bound we have:
\begin{align*}
\Pr\left[\left|\frac{t}{n - t}A q^t - \frac{t_0}{n - t_0} A q^{t_0}\right|_1 \ge \lambda \right] \le 4n e^{- \frac{\lambda^2}{2\sum_{i = 1}^t \frac{16 n^2}{(n - i)^2} }} \le 4 n e^{- \frac{\lambda^2 }{32 t}},
\end{align*}
where the bound on the step size of the martingale at time $t$ follows from Lemma~\ref{lem:martingale}.  
Let $w = \frac{\lambda (n - t)}{t}$. Then $\Pr\left[\left|Aq^t - \frac{t_0 (n - t)}{t (n - t_0)} A q^{t_0}\right|_1  \ge w \right] \le 4n e^{- \frac{w^2 t}{32 (n - t)^2}}$.
Thus: 
\begin{align*}
&\mathbb E\left[\left|Aq^t - \frac{t_0 (n - t)}{t (n - t_0)} A q^{t_0}\right|_1\right] = \int_0^{\infty} \Pr\left[\left|Aq^t - \frac{t_0 (n - t)}{t (n - t_0)} A q^{t_0}\right|_1 \ge w\right] dw \\
 &\le 4 n \int_0^\infty e^{- \frac{w^2 t}{32 (n - t)^2}} = 16n (n - t)\sqrt{\frac{\pi}{t}} = O\left(\frac{n (n - t)} {\sqrt{t}}\right) \qedhere
\end{align*}
In order to bound $\mathbb E [|A q^t|_1]$ note that $\mathbb E [|A q^t|_1] \le \mathbb E\left[\left|Aq^t - \frac{t_0 (n - t)}{t (n - t_0)} A q^{t_0}\right|_1\right] + \mathbb E\left[\left|\frac{t_0 (n - t)}{t (n - t_0)} A q^{t_0}\right|_1\right]$.
Bounding the second term we have:
\begin{align*}
\mathbb E\left[\left|\frac{t_0 (n - t)}{t (n - t_0)} A q^{t_0}\right|_1\right] = O\left(\frac{n - t}{\epsilon^2 t n}\right) \mathbb E\left[\left| A q^{t_0}\right|_1\right].
\end{align*}
It suffices to show that $\mathbb E\left[\left| A q^{t_0}\right|_1\right] = O(\epsilon n^2)$.
Indeed, consider $|A_{vi} q^{t_0}|_1 = \left|A_{vi} \left(\hat x^{t_0} - \frac{n}{t_0} x^{t_0}\right)\right|_1 =\left|A_{vi} \left(x^* - \frac{n}{t_0} x^{{t_0}}\right)\right|_1$. Recall that $x^{{t_0}}$ is defined as:
\begin{align*}
x^{t_0}_i = \begin{cases} x^*_i & \text{if } i \in S^{t_0}\\ 0, & \text{otherwise.} \end{cases}
\end{align*}
Random variable $\frac{n}{t_0}A_{vi} x^{t_0}$ has expectation $A_{vi} x^*$, by the Chernoff bound $\mathbb E\left[\left|A_{vi} \left(x^* - \frac{n}{t_0} x^{{t_0}}\right)\right|_1\right]  = O(\epsilon n)$.
\end{proofof}

\subsection{Proof of Lemma~\ref{lem:union-bound}}\label{app:union-bound}
\begin{proofof}{Lemma~\ref{lem:union-bound}}
To simplify presentation we observe that the following proposition (Lemma~2.20 in~\cite{S12}) can be used black-box here since it doesn't rely on the specific way the greedy choice is made in the definition of $x^t$ and $\hat x^t$ and only uses the fact that the variables at every step are chosen randomly together with the definition of a fictitious assignment. \todo{Explanation of why there is no induction in the proof as per Comment 5.}
\begin{proposition} [Adapted from Lemma 2.20 in\cite{S12}]
For every $t$ and a q-dimensional array $A$ with $|A|_\infty \le 1$ we have:
\begin{align*}
\Pr\left[\left|\myz{q^t}\right|\ge \sigma + \lambda\right] \le e^{-\frac{\lambda^2}{\sigma^2}},
\end{align*}
where $\sigma = O(n^q \sqrt {\frac{n - t}{nt}})$.
\end{proposition}

Using Lemma 2.20 from~\cite{S12} we have for every $y \in Y$:
\begin{align*}
\Pr\left[\left|\myb{q^t_{(y)}}\right| \ge \lambda\right] \le e^{- \frac{\lambda^2}{\sigma^2}}.
\end{align*}
The set $Y$ has size $2^{\frac{\log k}{\epsilon^2}}$, so by taking a union bound we have:
\begin{align*}
\Pr\left[\max_{y \in Y} |\myb{q^t_{(y)}}|_1 \ge \lambda\right] \le e^{O\left(\frac{\log k}{\epsilon^2}\right) - \lambda^2 / \sigma^2}
\end{align*}
The bound follows by integration over $\lambda$. This completes the proof of the first part of the lemma, the second part is identical.
\end{proofof}

%\section{Time-dependent sample size}\label{app:time-dependent-sampling}
%Instead of using uniform sample size in the algorithms above, one can introduce dependence on the time in the sample size. In Algorithm~\ref{alg:greedy-ptas-sampling} we can replace a sample of size $s = O\left(\frac{\log^2 n } {\epsilon^2}\right)$ with a sample of size $s_t = O\left(\frac{n^{2/3}}{\epsilon^2 t^{2/3}}\right)$ taken at time $t$. In fact, sampling at any rate $O\left(\frac{n^\delta}{\epsilon^2 t^\delta}\right)$ would work for $0 < \delta < 1$, the choice of $\delta = 2/3$ allows to minimize the constant factors.
%
%\begin{theorem}\label{thm:max-cut-greedy-sampling-time-dependent}
%A modification of Algorithm~\ref{alg:greedy-ptas-sampling} with time-dependent sample size $s_t = O\left(\frac{n^{2/3}}{\epsilon^2 t^{2/3}}\right)$ gives an additive $O(\epsilon n^2)$-approximation for \textsc{Max-Cut} in time $n \cdot 2^{O(1/\epsilon^2)}$
%\end{theorem}
%
%\begin{proof}
%For such choice of the sample size an analog of Lemma~\ref{lem:key-lemma} still holds.
%
%\begin{lemma} \label{lem:key-lemma-time-dependent}
%For all $t \ge t_0$ it holds that $\mathbb E \left[(\hat x^t)^T A \hat x^t\right] - \mathbb E\left[(\hat x^{t_0})^T A \hat x^{t_0}\right] \le O(\epsilon n^2).$
%\end{lemma}
%\end{proof}

\subsection{Proof of Lemma~\ref{lem:step-bound-rcsp}}\label{app:proof-step-bound-rcsp}
\begin{proof} \todo{Proof added for completeness, includes modified dependence on $k$.}

We use the following analog of Lemma~\ref{lem:step-bound-aux} for $r$-CSPs, whose proof can be adapted from~\cite{S12}.

\begin{lemma}\label{lem:step-bound-rcsp-aux} (Analogous to Lemma 2.13 in~\cite{S12}, analog of Lemma~\ref{lem:step-bound-aux})
For every $t$ it holds that:
\begin{align*}
A(\hat x^t, \dots, \hat x^t) - A (\hat x^{t - 1}, \dots, \hat x^{t - 1}) \le \frac{2^{r + 2} n^r}{t^r} + r A(\hat x^{t} - \hat x^{t - 1}, \hat x^{t - 1}, \dots, \hat x^{t - 1}).
\end{align*}
\end{lemma}

By Lemma~\ref{lem:step-bound-rcsp-aux} it suffices to bound $\mathbb E\left[A(\hat x^t - \hat x^{t - 1}, \hat x^{t - 1}, \dots, \hat x^{t -1})\right]$. We have:
\begin{align*}
A(\hat x^t - \hat x^{t - 1}, \hat x^{t - 1}, \dots, \hat x^{t - 1}) = \frac{n^{r - 1}}{t^{r - 1}} A(\hat x^t - \hat x^{t - 1}, x^{t - 1}, \dots, x^{t - 1}) + A(\hat x^{t} - \hat x^{t - 1}, q^{t - 1}, \dots, q^{t - 1}),
\end{align*}
using linearity of $A$ in each of its arguments and the definition $q^{t - 1} = \hat x^{t - 1} - \frac{n}{t - 1} x^{t - 1}$.
We bound the first term using Proposition~\ref{prop:sampling-concentration-rcsp}.
If $u  = r_t$ then $\hat x^t_u - \hat x^{t - 1}_u = g^t_u - \hat x^{t - 1}_u$.
Hence $\mathbb E\left[A(\hat x^t_{ui} - \hat x^{t - 1}_{ui}, x^{t - 1}. \dots, x^{t - 1})\right] \le \frac{2 k t^{r - 1}}{\sqrt{s_t}}$. Because there are $k$ different values of $i$ we have overall $\mathbb E\left[A(\hat x^t_{u} - \hat x^{t - 1}_{u}, x^{t - 1}. \dots, x^{t - 1})\right] \le \frac{2 k^2 t^{r - 1}}{\sqrt{s_t}}$.
If $u \notin S^t$ then we have $\hat x^t - \hat x^{t - 1} = \frac{1}{t}(g^t_u - \hat x^{t - 1}_u)$.
For such vertices we have $\mathbb E\left[A(\hat x^t_u - \hat x^{t - 1}_u, x^{t - 1}, \dots, x^{t -1})\right] \le \frac{2 k^2 t^{r - 1}}{\sqrt{s_t}}$. The total number of such vertices is $n - t$, so overall we have:
\begin{align*}
\mathbb E\left[A(\hat x^t - \hat x^{t -1}, x^{t - 1}, \dots, x^{t -1})\right] \le \frac{2 k^2 t^{t - 1}}{\sqrt{s_t}} + \frac{2 (n - t) k^2 t^{r - 2}}{\sqrt{s_t}} = \frac{2 k^2 n t^{r - 2}}{\sqrt{s_t}}.
\end{align*}
Thus, the first term is at most $\frac{2 k^2 n^r}{t \sqrt{s_t}}$.

The second term can be bounded using the same reasoning as in Lemma~\ref{lem:step-bound}. We give a proof for completeness. Conditioning on the choices of $S^{t - 1}$ we have:
\begin{align*}
&\mathbb E\left[A(\hat x^t - \hat x^{t - 1}_u, q^{t  - 1}, \dots, q^{t  - 1}) | S^{t - 1}\right] \\
 &= \sum_{v} \mathbb E \left[A(\hat x^t_v - \hat x^{t - 1}_v, q^{t - 1}, \dots, q^{t - 1}) | S^{t - 1}\right] \\
 & = \sum_v \mathbb E [\hat x^t_{v} - \hat x^{t - 1}_v | S^{t - 1}] \cdot A(e_u, q^{t-1}, \dots, q^{t - 1}) \\
 & \le \sum_v |\mathbb E[\hat x^t_{v} - \hat x^{t - 1}_v]|_1 \cdot |A(e_u, q^{t - 1}, \dots, q^{t - 1})|_1 \\
 & \le \frac{2 n}{t (n - t + 1)}\cdot |A(\cdot , q^{t - 1}, \dots, q^{t - 1})|_1,
\end{align*}
where the last inequality follows from the analysis given below.
For fixed $S^{t - 1}$ and $u$, if $u \in S^{t  - 1}$ then $\hat x^{t}_u - \hat x^{t - 1}_u = 0$.
Otherwise, as shown in the proof of Lemma~\ref{lem:step-bound-aux} it holds that $\sum_{u} |\hat x^t_u - \hat x^{t - 1}_u| \le \frac{2n}{t}$ and hence:
\begin{align*}
\mathbb E\left[|\hat x^t_u - \hat x^{t - 1}_u|_1|S^{t - 1}\right] = \mathbb E\left[\frac{|\hat x^t - \hat x^{t - 1}|_1}{n - t + 1} | S^{t - 1}\right] \le \frac{2 n}{t (n - t + 1)}.
\end{align*}
\end{proof}

\end{document}